\documentclass[a4paper]{scrartcl}

\usepackage{amssymb, amsthm, mathtools, cite}
\usepackage[hidelinks, citecolor=blue]{hyperref}
\hypersetup{colorlinks=true, urlcolor = blue, linkcolor=black}

\newtheorem{theorem}{Theorem}
\newtheorem{corollary}[theorem]{Corollary}

\newtheorem{assumption}[theorem]{Assumption}
\numberwithin{equation}{section}
\numberwithin{theorem}{section}

\newcommand{\rr}{{\mathbb{R}}}

\newcommand{\cp}{{\mathbb{C}_+}}
\newcommand{\ee}{{\mathbb{E}\,}}
\newcommand{\pp}{{\mathbb{P}}}
\newcommand{\oh}{{\mathcal{O}}}
\newcommand{\im}{{\operatorname{Im}\,}}
\newcommand{\re}{{\operatorname{Re }\,}}
\newcommand{\parder}[2]{{\frac{\partial #1}{\partial #2}}}
\newcommand{\supp}{{\operatorname{supp}\,}}
\newcommand{\tr}{{\operatorname{Tr}\,}}

\newcommand{\ket}[1]{{{|#1\rangle}}}
\newcommand{\bra}[1]{{{\langle #1|}}}
\newcommand{\ketbra}[1]{{{|#1\rangle\langle #1|}}}
\newcommand{\beq}[1]{\begin{equation} \label{#1}}
\newcommand{\eeq}{\end{equation}}
\newcommand{\stil}{{\tilde{S}}}
\newcommand{\gtil}{{\tilde{G}}}

\begin{document}
\addtokomafont{author}{\raggedright}
\title{\raggedright Non-Ergodic Delocalization in the Rosenzweig-Porter Model}
\author{\hspace{-.075in}Per von Soosten and Simone Warzel}
\date{\vspace{-.2in}}
\maketitle
\minisec{Abstract}
We consider the Rosenzweig-Porter model $H = V + \sqrt{T}\, \Phi$, where $V$ is a $N \times N$ diagonal matrix, $\Phi$ is drawn from the $N \times N$ Gaussian Orthogonal Ensemble, and $N^{-1} \ll T \ll 1$. We prove that the eigenfunctions of $H$ are typically supported in a set of approximately $NT$ sites, thereby confirming the existence of a previously conjectured non-ergodic delocalized phase. Our proof is based on martingale estimates along the characteristic curves of the stochastic advection equation satisfied by the local resolvent of the Brownian motion representation of $H$.
\bigskip

\section{Introduction}
This paper is concerned with the eigenfunctions of the Rosenzweig-Porter model~\cite{PhysRev.120.1698}, which is the simplest example of a random matrix with non-trivial spatial structure and provides a standard interpolation between integrability and chaos. Consisting of a rotationally invariant term and a potential, the model provides a highly simplified and analytically tractable toy model for the localization transition in disordered systems. Furthermore, the Rosenzweig-Porter model and its relatives, the Anderson models on the Bethe lattice and on the random regular graph, have recently received a renewed surge of interest related to the many-body localization transition~\cite{0295-5075-117-3-30003,0295-5075-115-4-47003,1367-2630-17-12-122002}. In that context, they provide basic examples of phases in which eigenfunctions delocalize over a large number of sites, but not uniformly over the entire volume.

The Hamiltonian in question is defined on $\ell^2(\{1,\dots,N\})$ by 
\beq{eq:hdef}H = V + \sqrt{T} \, \Phi,
\eeq
where $T > 0$, 
\[V = \sum_x V_x \ketbra{\delta_x}\]
is a sufficiently regular diagonal matrix with real entries $\{V_x\}$, and the Hermitian matrix $ \Phi $ has random  entries with normal distribution 
\[\langle \delta_y, \Phi \delta_x \rangle \sim \mathcal{N}\left( 0, \sqrt{\frac{1 + \delta_{xy}}{N}} \right)\]
which are independent up to the symmetry constraint. Here and throughout, $\delta_x \in \ell^2(\{1, \dots, N\})$ denotes the site basis element
\[\delta_x(u) = \begin{cases} 1 & \mbox{ if }  u = x\\ 0 & \mbox{ if } u \neq x\end{cases}\]
and $\delta_{xy} = \langle \delta_y, \delta_x \rangle$. We will follow Dyson's idea~\cite{MR0148397} and represent $H$ as a stochastic process
\[H_t = V + \Phi_t,\]
where the entries of $\Phi_t$ are normalized Brownian motions
\[\langle \delta_y, \Phi_t \delta_x \rangle = \sqrt{\frac{1 + \delta_{xy}}{N}} \, B_{xy}(t)\]
independent up to the symmetry constraint. Hence, $H_T$ has the same distribution as~\eqref{eq:hdef} and we consider only $H_T$ from now on. 

One expects the spectral behavior of $H_T$ to interpolate between $V$ and $\Phi_T$ as  $T$ increases. In terms of local eigenvalue statistics, recent works have established rigorously that there is a sharp transition at $T = N^{-1}$. On the one hand, if $T \gg N^{-1}$, Landon, Sosoe, and Yau ~\cite{landonyauarxiv,landonsosoeyauarxiv} proved that the local statistics fall into the Wigner-Dyson-Mehta universality class and agree asymptotically with those of the Gaussian Orthogonal Ensemble. On the other hand, if $T \ll N^{-1}$ and the $\{V_x\}$ are independent random variables, the local statistics converge to a Poisson point process and agree asymptotically with those of $V$~\cite{resflow}. Nevertheless, the present understanding of the transition  is incomplete with regard to the eigenfunctions of $H_T$. While it is known that the eigenfunctions are completely extended when $T \geq 1$~\cite{MR3134604}, and slightly weakened localization bounds were proved in~\cite{resflow} when $T \ll N^{-1}$, there are no previous rigorous results concerning the behavior of the eigenfunctions in the intermediate regime $N^{-1} \ll T \ll 1$. Moreover, the nature of the transition in the eigenfunctions in the related Anderson models on the Bethe lattice and random regular graph has been widely disputed even in the physics literature~\cite{PhysRevLett.113.046806, PhysRevLett.117.156601, PhysRevB.94.220203, biroliarxiv}. In the main result of this paper, we confirm the conjectured picture of Kravtsov et.~al.~\cite{1367-2630-17-12-122002} by proving that in the intermediate regime a normalized eigenfunction $\psi_\lambda$ corresponding to $\lambda \in \sigma(H_T)$ delocalizes across approximately those $NT \gg 1$ sites for which $V_x$ is closest to $\lambda$. This means that the mass of each eigenfunction spreads to a large number of sites. These sites nevertheless form a vanishing fraction of the entire volume $\{1, \dots, N\}$, indicating the existence of a non-ergodic delocalized phase. In~\cite{0295-5075-115-4-47003}, Facoetti, Vivo, and Biroli provided some clarifying analysis on the physics level of rigor in terms of second-order perturbation theory. Their work also explains how the abrupt transition in the local statistics does not contradict the gradual transition in the degree of eigenfunction localization, by arguing that the statistics retain a Poissonian character on mesoscopic scales greater than $T$. For rigorous results of a similar character concerning the eigenvalue statistics we refer to~\cite{duitsjohansson,landonhuangarxiv}.

Throughout this paper, we will fix a time $T = N^{-1 + \delta}$ with $\delta \in (0,1)$ and a spectral domain of the form
\[D = W + i[\eta, 1]\]
where $W \subset \rr$ is a bounded interval and $\eta = N^{-1 + \alpha}$ is a spectral scale whose parameter $\alpha > 0$ is fixed but may  be arbitrarily small. In what follows, we will require $V$ to possess some regularity that will be expressed in terms of the Stieltjes transform of the empirical eigenvalue measure
\[S_0(z) = \frac{1}{N} \sum_{x} \frac{1}{V_x - z}.\]
\begin{assumption}\label{thm:assump} There exist $\epsilon > 0$ and constants $K_m, K_l \in (0,\infty)$ such that
\begin{enumerate}
\item $|S_0(z)| \le K_m \log N$ uniformly in $\im z \geq \eta$ and
\item  $\im S_0(z) \ge K_l$ uniformly in $z \in \cp$ with $\im z \geq \eta$ and $\operatorname{dist}(z, D) \le \epsilon$.
\end{enumerate}
\end{assumption}

If the entries $V_x$ are drawn independently from a compactly supported density $\rho \in L^\infty$, we will show in Section \ref{sec:v} that Assumption \ref{thm:assump} is satisfied with asymptotically full probability for any interval $W$ on which $\rho$ is bounded below. The restriction of the conclusion of the following Theorem \ref{thm:nonergodicity} to the states in $W$ is then only a mild condition since $\rho$ coincides with the asymptotic density of states of $H_T$ (see, for example,~\cite{resflow}) and hence one expects that the majority of $\sigma(H_T)$ typically lies in $\supp \rho$.

\begin{theorem}\label{thm:nonergodicity} Let $\kappa > \delta > \theta$ and set
\[X_\lambda = \{x \in \{1, \dots, N\}: |\lambda - V_x| \le N^{-1 + \kappa}\}.\]
Then, there exists $\gamma > 0$ such that for any $p > 0$ and all sufficiently large $N$
\begin{enumerate}
\item The normalized eigenfunctions in $W$ carry only negligible mass outside $X_\lambda$:
\[\pp\left( \sup_{\lambda \in \sigma(H_T) \cap W} \sum_{x \in X_\lambda^c} |\psi_\lambda(x)|^2 > N^{-\gamma} \right) \le N^{-p}. \]
\item The normalized eigenfunctions in $W$ are maximally extended inside $X_\lambda$:
\[\pp\left( \sup_{\lambda \in \sigma(H_T) \cap W}  \|\psi_\lambda\|_\infty  > N^{-\theta/2} \right) \le N^{-p}.\]
\end{enumerate}
\end{theorem}

Theorem \ref{thm:nonergodicity} becomes meaningful when $\kappa$ and $\theta$ are chosen close to $\delta$. Under Assumption \ref{thm:assump}, the number of sites in $X_\lambda$ is then bounded by
\[|X_\lambda | \le 2N N^{-1 + \kappa} \, \im S_0(\lambda + iN^{-1 + \kappa}) \approx N^\kappa \log N \approx NT.\]
Moreover, the fact that
\[|\psi_\lambda(x)|^2 \le N^{-\theta} \approx (NT)^{-1}\]
shows that the eigenfunctions are maximally extended inside the subvolume $X_\lambda$.

One of the main novelties of this paper is the method of proof of the central resolvent estimates behind Theorem \ref{thm:nonergodicity}. Instead of applying the Schur complement formula and related arguments as initiated in the seminal work of Erd\H{o}s, Schlein, and Yau~\cite{MR2481753} (see also~\cite{MR3699468} and references therein), we control the local resolvent
\[G_t(x,z) = \langle \delta_x, R_t(z) \delta_x \rangle, \quad R_t(z) =  (H_t - z)^{-1}\]
by stochastic PDE methods. More specifically, it was shown in~\cite{resflow} that these quantities evolve according to the stochastic advection equation
\beq{eq:burgerseq} dG_t(x, z) = \left(S_t(z) \parder{}{z} G_t(x,z) + \frac{1}{2N} \parder{^2}{z^2} G_t(x,z) \right) \, dt + dM_t(x,z),
\eeq
where
\[S_t(z) = \frac{1}{N} \tr R_t(z)\]
is the normalized trace and $M_t(x,z)$ is a martingale given explicitly in terms of $R_t(z)$ below, cf.~\eqref{eq:martingal}. The key idea for the analysis of~\eqref{eq:burgerseq} is to track the evolution of $G_0(x, z)$ along the random characteristic curve of the nonlinear term defined by 
\beq{eq:chardef1} \dot{z}_t= -S_t(z_t).
\eeq
While the last two terms of~\eqref{eq:burgerseq} may change significantly in time on local scales $\im z \approx N^{-1}$, we will prove that their influence becomes negligible along the characteristic curves $z_t$. Up to finite-volume corrections, the curve $z_t$ therefore behaves like a true characteristic with $G_T(x, z_T) \approx G_0(x,z)$. Our bounds are strong enough to conclude that for every $z \in \cp$ with $\im z \gg N^{-1}$ there exists $w \in \cp$ with $\im w \approx T$ and $|w - z| = \oh(T)$ such that
\beq{eq:approxgrel2} G_T(x, z) \approx G_0(x,w).
\eeq
This means that the effect on the eigenfunctions of perturbing $V$ by $\Phi_T$ locally consists of a shift in the energy followed by a smearing of the scales below $T$. In essence, the change in the local resolvent on the given time scale is through an energy renormalization.

The relations~\eqref{eq:chardef1} and \eqref{eq:approxgrel2} amount to strong finite volume versions of the famous semi-circular flow of Pastur~\cite{MR0475502} localized to a single site and energy. Since the change in the resolvent is only by energy renormalization,~\eqref{eq:approxgrel2} is a quantitative local version of the subordination relations in free probability~\cite{MR1605393,MR1228526}. This perspective has been successfully pursued in the random matrix literature (cf.~\cite{MR3699468}). In particular, our resolvent bounds in Theorems~\ref{thm:martingalebound} and~\ref{thm:locresbound} can also be derived from a combination of results in~\cite{landonyauarxiv, MR3134604}. In comparison, our new proof is simple, self-contained, and avoids heavy machinery by only using well-known properties of martingales such as the reflection principle or the Burkholder-Davis-Gundy inequality. However, our strategy relies strongly on the Gaussian nature of the perturbation.

A powerful method for studying the eigenfunctions of $H_t$ directly was devised by Bourgade and Yau~\cite{MR3606475} and developed further by Bourgade, Huang, and Yau~\cite{bourgade2017}, whose Theorem 2.1 may also be used to derive the second point of Theorem \ref{thm:nonergodicity} above. This method was adapted to the present problem by Benigni~\cite{benignibourgade}, where it yields the local eigenvector statistics even for mesoscopic Wigner perturbations. This result covers Theorem \ref{thm:nonergodicity}, albeit only in a mesoscopic spectral window and with lower probability.

The paper is organized as follows. In Section \ref{sec:s}, we study the properties of the characteristic curves for~\eqref{eq:burgerseq}. In Section \ref{sec:g}, we bound the growth of the local resolvent along the characteristic curves and use this to prove Theorem \ref{thm:nonergodicity}. Finally, in Section \ref{sec:v}, we prove Assumption \ref{thm:assump} for random $\{V_x\}$.

\section{Characteristic Curves} \label{sec:s}
In this section, we study the properties of the characteristic curve
\beq{eq:characteristic}\dot{z}_t = -S_t(z_t), \qquad z_0 = z
\eeq
of the transport equation~\eqref{eq:burgerseq}. However, it is technically more convenient to consider instead the process
\[\xi_t(z) = z_{t \wedge \tau_z}\]
which is stopped at
\[\tau_z = \{\inf t > 0: \im z_t \le \eta/2\}.\]
Regarding $R_t(\xi_t(z))$ as a function of the processes $\{B_{uv}(t)\}$ and $\xi_t(z)$, It\^{o}'s lemma shows that
\begin{align*} dG_t(x, \xi_t(z))  &= \frac{1}{N} \sum_{u \le v} \langle \delta_x, R_t(\xi_t(z)) P_{uv} R_t(\xi_t(z)) P_{uv} R_t(\xi_t(z))\delta_x \rangle \, dt \\
&-  \frac{1}{\sqrt{N}}\sum_{u \le v} \langle \delta_x, R_t(\xi_t(z)) P_{uv} R_t(\xi_t(z)) \delta_x \rangle \, dB_{uv}(t) +  \dot{\xi}_t(z) \parder{}{\xi}G_t(x,\xi_t(z)) \, dt
\end{align*}
with 
\[P_{uv} = \frac{1}{\sqrt{1 + \delta_{uv}}} \left( \ket{\delta_u}\bra{\delta_v} + \ket{\delta_v}\bra{\delta_u}\right).  \] 
The piecewise $C^1$ process $\xi_t(z)$ has vanishing covariation with all the $B_{uv}(t)$. The calculations in the proof of Theorem 2.1 in~\cite{resflow} then show that
\begin{align*} dG_t(x, \xi_t(z))  &=  \left(S_t(\xi_t(z)) \parder{}{\xi} G_t(x,\xi_t(z)) + \frac{1}{2N} \parder{^2}{\xi^2} G_t(x,\xi_t(z)) \right) \, dt\\
& + \dot{\xi}_t(z)  \parder{}{\xi}G_t(x,\xi_t(z)) \, dt + dM_t(x,z)
\end{align*}
with
\beq{eq:martingal}
dM_t(x,z) =  -\frac{1}{\sqrt{N}}\sum_{u \le v} \langle \delta_x, R_t(\xi_t(z)) P_{uv} R_t(\xi_t(z)) \delta_x \rangle \, dB_{uv}(t).
\eeq
If $\tau$ is any stopping time such that $\tau \le \tau_z$ almost surely,~\eqref{eq:characteristic} yields
\begin{align}\label{eq:charito}G_\tau(x, \xi_\tau(z)) - G_0(x,z) &= \int_0^\tau \! \frac{1}{2N} \parder{^2}{\xi^2} G_t(x,\xi_t(z)) \, dt \nonumber\\ & -\frac{1}{\sqrt{N}}\sum_{u \le v} \int_0^\tau \! \langle \delta_x, R_t(\xi_t(z)) P_{uv} R_t(\xi_t(z)) \delta_x \rangle \, dB_{uv}(t)
\end{align}
for the change in the local resolvent along the characteristic curve.

Our next goal is to show that with high probability the change in $S_t$ along the curve $\xi_t(z)$ is small for a sufficiently dense set of initial points $z$. Let
\[\tilde{D} \subset \{z \in \cp: \im z > \eta\}\]
be some set of finite cardinality $| \tilde{D} | < \infty$. The next theorem bounds the probability of the event
\[\mathcal{A}_S = \left\{\sup_{z \in \tilde{D}} \sup_{t \le \tau_z}  |S_t(\xi_t(z)) - S_0(z)|  >  \frac{4}{\sqrt{N \eta}}  \right\},\]
showing that with high probability $S_t(\xi_t(z))$ is approximately constant if $|\tilde{D}|$ grows only polynomially in $N$. In the statement of the theorem, and throughout this paper, $C, c \in (0,\infty)$ denote deterministic constants that are independent of $N$ but whose value may change from instance to instance.

\begin{theorem} \label{thm:martingalebound} Let $\beta \in (0,1)$. For every $z \in \cp$ with $\im z > \eta$ we have
\beq{eq:mgboundformula} \pp\left(\sup_{t \le \tau_z} |S_t(\xi_t(z)) - S_0(z)| > \frac{4}{(N \eta)^\beta} \right) \le C \exp\left(-c (N\eta)^{2(1-\beta)}\right)
\eeq
and therefore $\pp\left(\mathcal{A}_S \right) \le C|\tilde{D}| e^{-c N\eta}$.
\end{theorem}
\begin{proof} By averaging~\eqref{eq:charito} over all $x$, we see that the process
\[\stil_t = S_{t \wedge \tau_z}(\xi_t(z))\]
satisfies
\begin{align*} \stil_t - \stil_0 &= \int_0^{t \wedge \tau_z} \! \frac{1}{2N} \parder{^2}{\xi^2} S_s(\xi_s(z)) \, ds \\
&-\frac{1}{\sqrt{N^3}}\sum_x \sum_{u \le v} \int_0^{t \wedge \tau_z} \! \langle \delta_x, R_s(\xi_s(z)) P_{uv} R_s(\xi_s(z)) \delta_x \rangle \, dB_{uv}(s).
\end{align*}
The drift component of $\stil$ is bounded by
\begin{align*}
\int_0^{t \wedge \tau_z} \! \frac{1}{2N} \left| \parder{^2}{\xi^2} S_s(\xi_s(z))\right| \, ds &\le \frac{1}{N}  \int_0^{t \wedge \tau_z}\frac{\im S_s(\xi_s(z))}{(\im \xi_s(z))^2}\, ds\\
&= \frac{1}{N}  \int_0^{t \wedge \tau_z}\frac{-d (\im \xi_s(z))}{(\im \xi_s(z))^2}\\
&= \frac{1}{N \, \im \xi_t(z)} - \frac{1}{N \, \im z} \le \frac{2}{N \eta}.
\end{align*}
The martingale part of $\stil$ is given by
\begin{align*}M_t &=  -\frac{1}{\sqrt{N^3}}\sum_x \sum_{u \le v} \int_0^{t \wedge \tau_z} \! \langle \delta_x, R_s(\xi_s(z)) P_{uv} R_s(\xi_s(z)) \delta_x \rangle \, dB_{uv}(s)\\
&= -\frac{1}{\sqrt{N^3}} \sum_{u, v}  \sqrt{1 + \delta_{uv}} \int_0^{t \wedge \tau_z} \! \langle \delta_v, R_s(\xi_s(z))^2 \delta_u\rangle \, dB_{uv}(s) \, . 
\end{align*}
Its quadratic variation may be expressed as
\begin{align*} [M]_t &\le  \frac{2}{N^3} \int_0^{t \wedge \tau_z} \!  \sum_{u, v} \left|\langle \delta_v, R_s(\xi_s(z))^2 \delta_u\rangle \right|^2 \, ds\\
&\le \frac{2}{N^2} \int_0^{t \wedge \tau_z} \!  \frac{\im S_s(\xi_s(z))}{ (\im \xi_s(z))^3} \, ds.\\
&=  \frac{1}{N^2}  \int_0^{t \wedge \tau_z} \! \frac{-2}{(\im \xi_s(z))^3} \, d(\im \xi_s(z))\\
&= \frac{1}{(N \, \im \xi_t(z))^2} - \frac{1}{(N \, \im z)^2} \le \frac{4}{(N \eta)^2}.
\end{align*}
It follows that there exists a Brownian motion $\tilde{B}$ such that
\[ \sup_t |\stil_t - \stil_0| \le \sup_t \left(\frac{2}{N \eta} + \left| \tilde{B}_{[M]_t} \right| \right) \le \frac{2}{(N \eta)^\beta} + \sup_{t \le 4/(N\eta)^2} |\tilde{B}_t|.\]
Applying the reflection principle to $\tilde{B}$ we obtain \eqref{eq:mgboundformula} and the second assertion follows from the union bound.
\end{proof}

The estimate on the fluctuations in~\eqref{eq:mgboundformula} is close to optimal regarding the power of $N\eta$. It can be inserted into the usual stability analysis and bootstrap, yielding the local deformed semicircle law with an optimal error bound without employing the intricate fluctuation averaging mechanism of~\cite{MR3119922}.
 
Once we know that $S_t(\xi_t(z))$ is approximately constant, this term can be inserted into integrals involving $\xi_t(z)$ more or less at will, and the substitution trick from Theorem \ref{thm:martingalebound} gives bounds improving on the trivial bound by a factor of $\eta$. We illustrate this in the following corollary, which will prove useful in extending our method to the local resolvents.

\begin{corollary} \label{thm:inserts} If $\mathcal{A}_S$ does not occur, then
\[ \int_0^{t \wedge \tau_z} \! \frac{1}{(\im \xi_s(z))^2} \, ds \le \frac{4}{K_l \eta}\]
for all $t > 0$ and $z \in D \cap \tilde{D}$.
\end{corollary}
\begin{proof}
If $\mathcal{A}_S$ does not occur and $z \in D \cap \tilde{D}$, then for sufficiently large $N$
\[\inf_{s \le t \wedge \tau_z} \im S_s(\xi_s(z)) \geq \im S_0(z) - \frac{4}{\sqrt{N\eta}} \geq \frac{K_l}{2}\]
where $K_l$ is the lower bound from Assumption \ref{thm:assump}. Hence,
\begin{align*} \int_0^{t \wedge \tau_z} \! \frac{1}{(\im \xi_s(z))^2} \, ds &\le \frac{2}{K_l} \int_0^{t \wedge \tau_z} \! \frac{\im S_s(\xi_s(z))}{(\im \xi_s(z))^2} \, ds\\
&= \frac{2}{K_l} \int_0^{t \wedge \tau_z} \! \frac{d(\im \xi_s(z))}{(\im \xi_s(z))^2} \le \frac{4}{K_l \eta}.
 \end{align*}
\end{proof}

The Picard-Lindel\"{o}f theorem and the Herglotz property of $S_s$ imply that, almost surely, for every $z \in D$ there exists a $w \in \cp$ with $\xi_T(w) = z$ satisfying the a-priori deterministic bound
\[|w - z| \le \int_0^T \! |S_s(\xi_s(w))| \, ds \le \frac{T}{\eta}.\]
In order for Theorem \ref{thm:martingalebound} to be useful in the study of the function $S_T$, we need to guarantee that a sufficiently dense subset of $D$ is of the form $\xi_T(w)$ with $w \in \tilde{D}$. To this end, we define the distance
\beq{eq:rdef}r = \min \left\{ \Upsilon T \eta^{2},  N^{-2\theta} \eta^3, N^{-(1+2\gamma)} \eta^{3}\right\},
\eeq
where
\beq{eq:kdef}\Upsilon = \sup_{\im z > \eta} |S_0(z)| + \frac{4}{\sqrt{N\eta}} \le C \log N,
\eeq
and $\gamma, \theta > 0$ are the parameters from the statement of Theorem \ref{thm:nonergodicity}. We now require $\tilde{D}$ to be such that
\[\operatorname{dist}(z, \tilde{D}) \le r\]
for all $z \in \cp$ with $\im z > \eta$ and $\operatorname{dist}(z, D) \le T/\eta$. The grid $\tilde{D}$ can hence be chosen such that its cardinality is bounded by
\[|\tilde{D}| \le C(\eta r)^{-2}.\]
The following theorem provides a Lipschitz constant for the characteristic flow which grows only polynomially in $\eta$. The resulting bound is a significant improvement on the exponential bound provided by the direct application of Gr\"{o}nwall's inequality and enables us to keep the cardinality of $\tilde{D}$ polynomial in $N$.

\begin{theorem} \label{thm:comparisonpoint} Suppose $\mathcal{A}_S$ does not occur and $N$ is sufficiently large. Then for every $z \in D$ there exists $w \in  \tilde{D}$ such that:
\begin{enumerate}
\item $ \displaystyle |\xi_T(w) - z| \le  C \eta^{-2} r $,
\item  $ \displaystyle |w - z| \le C\Upsilon T \ $ with $ \Upsilon $ as in~\eqref{eq:kdef}, and
\item $\im w \geq \frac{1}{2}K_l T$ with $K_l$ defined in Assumption \ref{thm:assump}.
\end{enumerate}
\end{theorem}
\begin{proof} By the construction of $\tilde{D}$, for any $ z \in D $ there exist $w_0 \in \cp$ with $\xi_T(w_0) = z$ and $w \in \tilde{D}$ with $|w - w_0| \le r$. If $t \le \tau_{{w_0}} \wedge \tau_{w}$, the evolution~\eqref{eq:characteristic} yields
\begin{align*} |\xi_t(w_0) - \xi_t(w)| &\le |{w_0} - w| +  \int_0^{t} \! |S_s(\xi_s(w_0)) - S_s(\xi_s(w))| \, ds\\
&\le |{w_0} - w| +  \frac{1}{N} \int_0^{t} \! \sum_{i} \frac{|\xi_s(w_0) - \xi_s(w)|}{|\lambda_i(s) - \xi_s(w_0)| |\lambda_i(s) - \xi_s(w)|}\, ds.
\end{align*}
Using the inequality $ab \le \frac{1}{2}(a^2 + b^2)$, the integral in the last term is bounded by
\begin{align*}
& \frac{1}{2} \int_0^{t} \!  |\xi_s(w_0) - \xi_s(w)| \frac{1}{N} \sum_{i} \left( \frac{1}{|\lambda_i(s) - \xi_s(w_0)|^2} + \frac{1}{ |\lambda_i(s) - \xi_s(w)|^2} \right)\, ds\\
&\le  \frac{1}{2} \int_0^{t} \!  |\xi_s(w_0) - \xi_s(w)| \left( \frac{\im S_s(\xi_s(w_0))}{\im \xi_s(w_0)} + \frac{\im S_s(\xi_s(w))}{\im \xi_s(w)} \right)\, ds
\end{align*}
so Gr\"{o}nwall's inequality shows that
\begin{align*} \log \frac{|\xi_t(w_0) - \xi_t(w)|}{|{w_0} - w|} &\le \frac{1}{2}\int_0^{t} \!  \frac{\im S_s(\xi_s(w_0))}{\im \xi_s(w_0)} +  \frac{\im S_s(\xi_s(w))}{\im \xi_s(w)} \, ds\\
&= - \frac{1}{2}\int_0^{t} \!  \frac{ d(\im \xi_s(w_0))}{\im \xi_s(w_0)} - \frac{1}{2}\int_0^{t} \! \frac{d( \im \xi_s(w))}{\im \xi_s(w)}\\
&= \log \sqrt{\frac{\im {w_0}}{\im \xi_t(w_0)} \frac{\im w} {\im \xi_t(w)}}.
\end{align*}
Thus, using $ \im w \leq \im w_0 + r  $, $ \im w_0 \leq 1 + T \eta^{-1} \leq C \eta^{-1} $, and the stopping rules, we obtain
\beq{eq:lipschitzflow} |\xi_t(w_0) - \xi_t(w)| \le C \eta^{-2} |{w_0} - w| \le C \eta^{-2} r
\eeq
for all $t \le \tau_{{w_0}} \wedge \tau_w$. Since $\im \xi_T(w_0) = z$ and $\im \xi_t(w_0)$ is decreasing, $\tau_{{w_0}} > T$, so~\eqref{eq:lipschitzflow} and the definition of $r$ shows that for sufficiently large $N$ we have $|\xi_t(w_0) - \xi_t(w)| \le \eta/4$ for all $t \le T \wedge \tau_w$. If it were true that $\tau_w < T$, we would obtain the contradiction
\[\frac{\eta}{2} = \im \xi_{\tau_w}(w) \geq \im \xi_{\tau_w}(w_0) - \frac{\eta}{4} \geq \eta - \frac{\eta}{4}.\]
Hence~\eqref{eq:lipschitzflow} is valid for $t = T$, establishing the first claim of the theorem. If $\mathcal{A}_S$ does not occur, then
\[|\xi_T(w) - w| \le \int_0^T \! |S_s(\xi_s(w))| \, ds \le \int_0^T \! |S_0(w)| + \frac{4}{\sqrt{N\eta}} \, ds \le \Upsilon T\]
since $w \in \tilde{D}$. Hence the definition of $r$ and~\eqref{eq:lipschitzflow} yield
\[ |w - z| \le |w - \xi_T(w)| + |\xi_T(w)- z| \le \Upsilon T+ C\Upsilon T = C\Upsilon T.\]
proving the second claim of the theorem. The second claim also implies that $\operatorname{dist}(w, D) \le \epsilon$ for sufficiently large $N$ so that Assumption \ref{thm:assump} guarantees $\im S_0(w) \geq K_l$. On the complement of $\mathcal{A}_S$ this yields
\[\im w = \im \xi_T(w) + \int_0^T \! \im S_s(\xi_s(w)) \, ds \geq  T\left(\im S_0(w) - \frac{4}{\sqrt{N \eta}}\right) \geq \frac{K_l T}{2}\]
for sufficiently large $N$.
\end{proof}

\section{Local Resolvent Bounds} \label{sec:g}
Since $S_t(z)$ is entirely featureless regarding a possible localization transition when $\im z \gg N^{-1}$, we now turn our attention to controlling the local resolvents $G_t(x,z)$ along the characteristic $\xi_t(z)$. Unlike $S_t$, the function $G_t(x,\cdot)$ may be heavily concentrated around certain energies in non-ergodic regimes. Therefore, its derivative may be large in all directions and we cannot expect an exact analogue of Theorem \ref{thm:martingalebound} to hold true for all energies. However, one may hope that the change in $G_t(x, z)$ along the characteristic is small in those regions where $G_t(x,z)$ itself is small. We encode this phenomenon in the event
\[\mathcal{A}_G(\ell) = \left\{\sup_{x} \sup_{z \in D \cap \tilde{D}} \sup_{s \le \tau_z} \frac{\im G_s(x,\xi_s(z))}{\im G_0(x,z)} > N^\ell \right\},\]
whose probability does in fact decay as $N \to \infty$. The proof is somewhat reminiscent of a Gr\"{o}nwall-type argument for martingales, which is greatly facilitated by the built-in control of the running maximum. Still, the basic mechanism behind the following argument is different from the stochastic Gr\"{o}nwall lemmas that previously appeared in~\cite{MR2718125,MR3078830}.

\begin{theorem} \label{thm:locresbound} For every $\ell > 0$ and $p > 0$ we have
\[\pp\left(\mathcal{A}_G(\ell) \right) \le N^{-p} \]
for all sufficiently large $N$.
\end{theorem}
\begin{proof} Fix $z \in D \cap \tilde{D}$ and consider the stopping time
\[\tau = \tau_z \wedge \inf \left\{ t \geq 0:  \int_0^{t \wedge \tau_z} \! \frac{1}{(\im \xi_s(z))^2} \, ds \geq \frac{5}{K_l \eta}
\right\},\]
where $K_l$ is the lower bound from Assumption \ref{thm:assump}. As in the proof Theorem \ref{thm:martingalebound}, the stopped process $\gtil_t = G_{t \wedge \tau}(x, \xi_{t \wedge \tau}(z))$ satisfies
\begin{align*} \gtil_t - \gtil_0 &= \int_0^{t \wedge \tau} \! \frac{1}{2N} \parder{^2}{\xi^2}  G_s(x, \xi_s(z)) \, ds \\
&  - \frac{1}{\sqrt{N}} \sum_{u \le v} \int_0^{t \wedge \tau} \! \langle \delta_x, R_s(\xi_s(z)) P_{uv} R_s(\xi_s(z)) \delta_x \rangle \, dB_{uv}(s).
\end{align*}
The drift component of $\gtil$ is bounded by
\begin{align*} \int_0^{t \wedge \tau} \! \frac{1}{2N} \left| \parder{^2}{\xi^2}  G_s(x, \xi_s(z))\right| \, ds &\le \left( \sup_{s \le T} \im \gtil_s \right) \int_0^{T \wedge \tau} \! \frac{1}{N(\im \xi_s(z))^2} \, ds \\
&\le \frac{5}{K_l N\eta}\left( \sup_{s \le T} \im \gtil_s \right),
\end{align*}
and, letting $M$ denote the martingale part of $\gtil$, its quadratic variation is bounded as follows,
\begin{align*} [M]_T &\le \frac{2}{N} \int_0^{T \wedge \tau} \! \sum_{u, v} \left| \langle \delta_x, R_s(\xi_s(z)) \delta_u \rangle \langle \delta_v, R_s(\xi_s(z)) \delta_x \rangle \right|^2 \, ds \nonumber\\
&= \frac{2}{N} \int_0^{T \wedge \tau}\! \left(\sum_{u} \left| \langle \delta_x, R_s(\xi_s(z)) \delta_u \rangle \right|^2\right) \left(\sum_{v} \left| \langle \delta_v, R_s(\xi_s(z)) \delta_x \rangle \right|^2\right) \, ds \nonumber\\
&= \frac{2}{N} \int_0^{T \wedge \tau} \! \left(\frac{\im G_s(x,\xi_s(z))}{\im \xi_s(z)}\right)^2 \, ds\\
&\le \left( \sup_{s \le T \wedge \tau} \im G_s(x,\xi_s(z)) \right)^2 \int_0^{T \wedge \tau} \!\frac{2}{N(\im \xi_s(z))^2} \, ds\\
&\le \frac{10}{K_l N\eta} \left( \sup_{s \le T} \im \gtil_s \right)^2.
\end{align*}
Hence,
\[\sup_{s \le T} \im \gtil_s \le \im \gtil_0 +\frac{5}{K_lN\eta} \sup_{s \le T} \im \gtil_s + \sup_{s \le T} |M_s|\]
so the Burkholder-Davis-Gundy inequality (with exponent $ q > 0 $ and constant $ C_q $) yields
\begin{align*}\left(1 -  \frac{5}{K_l N\eta}  \right)\left( \ee \left|\sup_{s \le T} \im \gtil_s\right|^q \right)^{1/q}  &\le \im  \gtil_0 +  \left( \ee \left| \sup_{s \le T} |M_{s}| \right|^q \right)^{1/q}\\
&\le \im  \gtil_0 + C_q  \left( \ee  [M]_T^{q/2} \right)^{1/q}\\
&\le \im \gtil_0  +  C_q \sqrt{\frac{10}{K_lN\eta}}  \left( \ee \left|\sup_{s \le T} \im \gtil_s\right|^q \right)^{1/q}.
\end{align*}
Since $N\eta \to \infty$, we can choose $N$ large enough such that $(1+C_q)\sqrt{\frac{10}{K_lN\eta}} < 1/2$. Rearranging and applying Markov's inequality shows
\[\pp\left( \sup_{s \le T \wedge \tau} \im G_s(x,\xi_s(z)) > 4N^\ell \im G_0(x,z)\right) \le N^{-\ell q}.\]
By Corollary \ref{thm:inserts}, $\tau = \tau_z$ on the event $\mathcal{A}_S^c$ and we conclude that
\[\pp\left( \sup_{s \le T \wedge \tau_z} \im G_s(x,\xi_s(z)) > 4N^\ell \im G_0(x,z)\right) \le N^{-\ell q} + \pp(\mathcal{A}_S),\]
so, choosing $q$ large enough, the theorem follows from the union bound.
\end{proof}

To prove Theorem \ref{thm:nonergodicity}, it remains only to combine the previous results with the fact that $G_T(x, \cdot)$ is the Stieltjes transform of the spectral measure at $x$.

\begin{proof}[Proof of Theorem \ref{thm:nonergodicity}]
We now specify the parameters $\alpha, \gamma, \ell > 0$ occuring in the spectral scale $\eta = N^{-1+\alpha}$, the definition of $r$ in~\eqref{eq:rdef}, and the event $\mathcal{A}_G(\ell)$ of Theorem \ref{thm:locresbound} by requiring that
\[0 < \gamma < \kappa - (\alpha + \ell + \delta).\]
Suppose that neither of the events $\mathcal{A}_S, \mathcal{A}_G(\ell)$ of Theorems \ref{thm:martingalebound} and \ref{thm:locresbound} occur, which is the case with probability $1 - N^{-p}$ provided $N$ is sufficiently large. For every $\lambda \in \sigma(H_T) \cap W$,
\[\sum_{x \in X_\lambda^c} |\psi_\lambda(x)|^2 \le \sum_{x \in X_\lambda^c} \sum_{E \in \sigma(H_T)} \frac{\eta^2}{(E - \lambda)^2 +  \eta^2} |\psi_E(x)|^2 = \eta \sum_{x \in X_\lambda^c}\im G_T(x,z)\]
with $z = \lambda + i\eta$. By Theorem \ref{thm:comparisonpoint}, there exists $w \in \tilde{D}$ such that  $|w - z| \le C\Upsilon T$, $\im w > K_l T/2$, and $|\xi_T(w) - z| \le C\eta N^{-(1+2\gamma)}$. Hence, for sufficiently large $N$,
\[\re w \in I \vcentcolon= [\lambda - N^{-1 + \kappa}, \lambda + N^{-1 + \kappa}], \quad \operatorname{dist}(\re w, \partial I) > d_N \vcentcolon=  \frac{1}{2}N^{-1 + \kappa}.\]
By the $\eta^{-2}$-Lipschitz continuity of $G_T(x, z)$, this yields
\begin{align*}\sum_{x \in X_\lambda^c} |\psi_\lambda(x)|^2 &\le \eta \sum_{x \in X_\lambda^c}\im G_T(x,\xi_T(w)) + CN^{-2\gamma}\\
&\le \eta N^\ell \sum_{x \in X_\lambda^c} \im G_0(x,w) + CN^{-2\gamma}.
\end{align*}
The sum may be bounded by recalling that $|V_x- \lambda| > d_N$ when $x \in X_\lambda^c$ and estimating
\begin{align*} \frac{1}{N} \sum_{x \in X_\lambda^c} \im G_0(x,w) &\le \frac{2}{N} \sum_{x \in X_\lambda^c} \frac{\im w}{(V_x - \re w)^2 + d_N^2}\\
&\le \frac{2 \, \im w}{d_N} \im S_0(\re w + i d_N)\\
&\le 4K_m (\im w) N^{1-\kappa} \log N,
\end{align*}
where we employed Assumption~\ref{thm:assump} in the last line. Since $\im w \le \eta + \Upsilon T\le CN^{-1 + \delta} \log N$, these observations combine to show that
\[\sum_{x \in X_\lambda^c} |\psi_\lambda(x)|^2 \le C N^{\alpha + \ell + \delta - \kappa} (\log N)^2 + C N^{-2\gamma} \le N^{-\gamma},\]
proving the first claim of the theorem. 

If, in addition, we require that $\alpha + \ell < \delta - \theta$, the second claim follows by the same token. Combining the Lipschitz continuity of $G_T(x,z)$ with $|\xi_T(w) - z| \le C\eta N^{-2\theta}$, we obtain
\begin{align*}
|\psi_\lambda(x)|^2 &\le \eta \, \im G_T(x, \lambda + i\eta)\\
&\le \eta \, \im G_T(x, \xi_T(w)) + CN^{-2\theta}\\
&\le \eta \, N^\ell |G_0(x, w)| + CN^{-2\theta}\\
&\le CN^{\alpha + \ell - \delta } + CN^{-2\theta} \le N^{-\theta}
\end{align*}
since $\im w > K_lT/2$.
\end{proof}

\section{Regularity Estimates for Random Potentials} \label{sec:v}
This section is devoted to the verification of Assumption \ref{thm:assump} in the case that the $\{V_x\}$ are drawn independently from a compactly supported density $\rho \in L^\infty$. We will assume that $\rho$ is bounded below in a neighborhood of $W$, i.e. there exists $\epsilon > 0$ such that
\[ \inf_{v \in W(\epsilon)} \rho(v) > 0\]
with $W(\epsilon) = W +[-\epsilon, \epsilon]$. We start by proving a concentration inequality in the spirit of Cram\'{e}r's theorem for $S_0$, which is uniform in spectral domains of the form
\[D(J, \zeta) = \left\{z \in \cp: \re z \in J \mbox{ and } \zeta \le \im z \le 1\right\}.\]

\begin{theorem} \label{thm:ficoncentration} Let $J \subset \rr$ be bounded. Then
\[ \pp\left(\sup_{z \in D(J,\zeta)} |\im S_0(z)- \ee \im S_0(z)| > \mu \right) \le C |J| \mu^{-2} \zeta^{-4}e^{-c\mu\sqrt{N\zeta}}\]
for all $\mu > 0$.
\end{theorem}
\begin{proof} Let $z = \alpha + i\beta$. Performing the substitution $v = (\tilde{v}-\alpha)/\beta$, we obtain
\begin{align*} \ee e^{t \im S_0(z)} &=  \left(\beta \int \! \rho(\alpha  + \beta v) \exp\left(\frac{t}{N \beta } \frac{1}{1 + v^2} \right) \, dv \right)^N\\
&\le \left(1 + \frac{t \ee \im S_0(z)}{N}  + \frac{t^2 \|\rho\|_\infty}{N^2 \beta}  \int \! \left(\frac{1}{1 + v^2}\right)^2 \exp\left(\frac{t}{N \beta } \frac{1}{1 + v^2} \right) \, dv \right)^N
\end{align*}
by Taylor's theorem. We choose $t = \sqrt{N\beta}$. Since $(1 + v^2)^{-2} \in L^1$ and
\[\frac{t}{N \beta } \frac{1}{1 + v^2}  \le \sqrt{2},\]
there exists an absolute constant $C < \infty$ such that
\begin{align*} \ee e^{t \im S_0(z)} &\le \left(1 + \frac{t \ee \im S_0(z)}{N} + C \left(\frac{t}{N\beta}\right)^2 \beta \right)^N\\
&\le \exp\left(N\left( \frac{t \ee \im S_0(z)}{N} + C \left(\frac{t}{N\beta}\right)^2 \beta \right) \right)\\
&= \exp(t\ee \im S_0(z)) \exp\left(\frac{Ct^2}{N\beta}\right).
\end{align*}
Using an exponential Chebyshev argument, we conclude that
\begin{align*}
\pp\left(\im S_0(z) \geq \ee \im S_0(z) + \mu \right) &\le e^{-t(\ee \im S_0(z) + \mu)} \ee e^{t \im S_0(z)}\\
&\le e^{-t\mu} \exp\left(\frac{Ct^2}{N\beta}\right)\\
&\le C \,e^{- c\mu \sqrt{N \zeta}}.
\end{align*}
The proof of the lower bound works the same way. Replacing the previous Chebyshev bound with
\[\pp\left(\im S_0(z) \le \ee \im S_0(z) - \mu \right) \le e^{t(\ee \im S_0(z) - \mu)} \ee e^{-t \im S_0(z)},\]
yields that for every fixed $z \in D(J,\zeta)$
\[ \pp\left(|\im S_0(z) - \ee \im S_0(z)| > \mu \right) \le C\, e^{-c\mu\sqrt{N\zeta}}.\]
Since $D(J,\zeta)$ is bounded, there exists a set of at most $C|J| \mu^{-2} \zeta^{-4}$ points $\{z_k\} \subset D(J,\zeta)$ such that for every $z \in D(J,\zeta)$ there exists $k$ with $|z - z_k| \le \frac{\mu\zeta^2}{12}$. By the union bound,
\[ \pp\left(\sup_k |\im S_0(z_k) - \ee \im S_0(z_k)| > \frac{\mu}{3} \right) \le C |J| \mu^{-2} \zeta^{-4} e^{-c\mu\sqrt{N\zeta}}.\]
But $\im S_0$ and $\ee \im S_0$ are $(2/\zeta)^{2}$-Lipschitz continuous in $D(J,\zeta)$ and thus
\begin{align*} |\im S_0(z) - \ee\im  S_0(z)| &\le |\im S_0(z) -\im S_0(z_k)| + |\im S_0(z_k) - \ee \im S_0(z_k)|\\ 
&+ |\ee \im S_0(z_k) - \ee \im S_0(z)| \le \mu,
\end{align*}
extending the bound to all $z \in D(J,\zeta)$.
\end{proof}

Since $\rho$ was assumed to be bounded below in $W(\epsilon)$, the corresponding lower bound for $\im S_0$ in the $\epsilon$-fattening of the original spectral domain $D$ follows immediately, proving the second point in Assumption \ref{thm:assump}.

\begin{corollary} There exists $K_l \in (0, \infty)$ such that
\[ \pp\left( \inf \left\{ \im S_0 (z): \im z \geq \eta \mbox{ and } \operatorname{dist}(z, D) \le \epsilon \right\} < K_l \right) \le C \eta^{-4}e^{-c\sqrt{N\eta}}.\]
\end{corollary}

To finish the proof of Assumption~\ref{thm:assump}, we combine the previous estimates with a standard argument for the Hilbert transform to produce a logarithmic bound for $|S_0|$.

\begin{corollary}\label{thm:kbound} There exists  $K_m \in (0, \infty)$ such that
\[ \pp\left(\sup_{\im z > \eta} |S_0 (z)| > K_m  + K_m \log \left(1 + \eta^{-2} \right) \right) \le C \eta^{-4}e^{-c\sqrt{N\eta}}.\]
\end{corollary}
\begin{proof} Using that $\rho$ is compactly supported and the trivial estimate
\[ \im S_0(z) \le \frac{1}{\operatorname{dist}(z, \supp \rho)},\]
Theorem \ref{thm:ficoncentration} with $J = \supp \rho + [-1,1]$ and $\zeta = \eta/2$ shows that there exist  $C,c,K_m \in (0, \infty)$ such that
\[ \pp\left(\sup_{\im z > \frac{\eta}{2}} \im S_0(z) > K_m \right) \le C \eta^{-4}e^{-c\sqrt{N\eta}}.\]
Letting
\[Q_z(t) = \frac{1}{\pi}\frac{t-\re z}{(t - \re z)^2 + (\im z)^2}\]
be the conjugate Poisson kernel and writing $z = \alpha + i(\beta/2)$, we see that on the complement of this event
\begin{align*} \re S_0(\alpha + i\beta) &= \int \! \im S_0\left(t - z \right) Q_{i\frac{\beta}{2}}(t) \, dt\\
&= \int_{[-1,1]} \! \im S_0\left(t -z\right) Q_{i\frac{\beta}{2}}(t) \, dt +  \int_{\rr \setminus [-1,1]} \! \im S_0\left(t - z\right) Q_{i\frac{\beta}{2}}(t) \, dt\\
&\le K_m \frac{1}{\pi} \int_{-1}^1 \, \frac{|t|}{t^2 + \beta^2} \, dt +  \frac{1}{\pi} \int \! \im S_0\left(t + i\frac{\beta}{2}\right) \, dt\\
&\le K_m \log \left(1 + \beta^{-2} \right) + 1.
\end{align*}
\end{proof}

\subsection*{Acknowledgments}
It is a pleasure to thank L.~Benigni for interesting discussions about the Rosenzweig-Porter model during the PCMI Summer Session (funded by NSF grant DMS:1441467). We also thank an anonymous referee for removing a superfluous assumption from an earlier version of this paper. This work was supported by the DFG (WA 1699/2-1).
\bibliographystyle{abbrv}
\bibliography{References}

\begin{thebibliography}{10}

\bibitem{PhysRevLett.117.156601}
B.~L. Altshuler, E.~Cuevas, L.~B. Ioffe, and V.~E. Kravtsov.
\newblock Nonergodic phases in strongly disordered random regular graphs.
\newblock {\em Phys. Rev. Lett.}, 117:156601, Oct 2016.

\bibitem{0295-5075-117-3-30003}
M.~Amini.
\newblock Spread of wave packets in disordered hierarchical lattices.
\newblock {\em EPL (Europhysics Letters)}, 117(3):30003, 2017.

\bibitem{benignibourgade}
L.~Benigni.
\newblock Eigenvectors distribution and quantum unique ergodicity for deformed
  {W}igner matrices.
\newblock Preprint available at ar{X}iv:1711.07103, 2017.

\bibitem{MR1605393}
P.~Biane.
\newblock Processes with free increments.
\newblock {\em Math. Z.}, 227(1):143--174, 1998.

\bibitem{biroliarxiv}
G.~Biroli, A.~C. Ribeiro-Teixera, and M.~Tarzia.
\newblock Difference between level statistics, ergodicity and localization
  transitions on the {B}ethe lattice.
\newblock Preprint available at ar{X}iv:1211.7334, 2012.

\bibitem{bourgade2017}
P.~Bourgade, J.~Huang, and H.-T. Yau.
\newblock Eigenvector statistics of sparse random matrices.
\newblock {\em Electron. J. Probab.}, 22:38 pp., 2017.

\bibitem{MR3606475}
P.~Bourgade and H.-T. Yau.
\newblock The eigenvector moment flow and local quantum unique ergodicity.
\newblock {\em Comm. Math. Phys.}, 350(1):231--278, 2017.

\bibitem{PhysRevLett.113.046806}
A.~De~Luca, B.~L. Altshuler, V.~E. Kravtsov, and A.~Scardicchio.
\newblock Anderson localization on the {B}ethe lattice: Nonergodicity of
  extended states.
\newblock {\em Phys. Rev. Lett.}, 113:046806, Jul 2014.

\bibitem{duitsjohansson}
M.~Duits and K.~Johansson.
\newblock On mesoscopic equilibrium for linear statistics in {D}yson's
  {B}rownian motion.
\newblock {\em To appear in Memoirs of the American Mathematical Society},
  2017.

\bibitem{MR0148397}
F.~J. Dyson.
\newblock A {B}rownian-motion model for the eigenvalues of a random matrix.
\newblock {\em J. Mathematical Phys.}, 3:1191--1198, 1962.

\bibitem{MR3119922}
L.~Erd\H{o}s, A.~Knowles, and H.-T. Yau.
\newblock Averaging fluctuations in resolvents of random band matrices.
\newblock {\em Ann. Henri Poincar\'e}, 14(8):1837--1926, 2013.

\bibitem{MR2481753}
L.~Erd\H{o}s, B.~Schlein, and H.-T. Yau.
\newblock Local semicircle law and complete delocalization for {W}igner random
  matrices.
\newblock {\em Comm. Math. Phys.}, 287(2):641--655, 2009.

\bibitem{MR3699468}
L.~Erd\H{o}s and H.-T. Yau.
\newblock {\em A dynamical approach to random matrix theory}, volume~28 of {\em
  Courant Lecture Notes in Mathematics}.
\newblock Courant Institute of Mathematical Sciences, New York; American
  Mathematical Society, Providence, RI, 2017.

\bibitem{0295-5075-115-4-47003}
D.~Facoetti, P.~Vivo, and G.~Biroli.
\newblock From non-ergodic eigenvectors to local resolvent statistics and back:
  A random matrix perspective.
\newblock {\em EPL (Europhysics Letters)}, 115(4):47003, 2016.

\bibitem{landonhuangarxiv}
J.~Huang and B.~Landon.
\newblock Local law and mesoscopic fluctuations of {D}yson {B}rownian motion
  for general $\beta$ and potential.
\newblock {\em arXiv:1612.06306}, 2017.

\bibitem{1367-2630-17-12-122002}
V.~E. Kravtsov, I.~M. Khaymovich, E.~Cuevas, and M.~Amini.
\newblock A random matrix model with localization and ergodic transitions.
\newblock {\em New Journal of Physics}, 17(12):122002, 2015.

\bibitem{landonsosoeyauarxiv}
B.~Landon, P.~Sosoe, and H.-T. Yau.
\newblock Fixed energy universality for {D}yson {B}rownian motion.
\newblock Preprint available at arXiv:1609.09011, 2016.

\bibitem{landonyauarxiv}
B.~Landon and H.-T. Yau.
\newblock Convergence of {L}ocal {S}tatistics of {D}yson {B}rownian {M}otion.
\newblock {\em Comm. Math. Phys.}, 355(3):949--1000, 2017.

\bibitem{MR3134604}
J.~O. Lee and K.~Schnelli.
\newblock Local deformed semicircle law and complete delocalization for
  {W}igner matrices with random potential.
\newblock {\em J. Math. Phys.}, 54(10):103504, 62, 2013.

\bibitem{MR0475502}
L.~A. Pastur.
\newblock The spectrum of random matrices.
\newblock {\em Teoret. Mat. Fiz.}, 10(1):102--112, 1972.

\bibitem{PhysRev.120.1698}
N.~Rosenzweig and C.~E. Porter.
\newblock "{R}epulsion of energy levels" in complex atomic spectra.
\newblock {\em Phys. Rev.}, 120:1698--1714, Dec 1960.

\bibitem{MR3078830}
M.~Scheutzow.
\newblock A stochastic {G}ronwall lemma.
\newblock {\em Infin. Dimens. Anal. Quantum Probab. Relat. Top.},
  16(2):1350019, 4, 2013.

\bibitem{PhysRevB.94.220203}
K.~S. Tikhonov, A.~D. Mirlin, and M.~A. Skvortsov.
\newblock Anderson localization and ergodicity on random regular graphs.
\newblock {\em Phys. Rev. B}, 94:220203, Dec 2016.

\bibitem{MR1228526}
D.~Voiculescu.
\newblock The analogues of entropy and of {F}isher's information measure in
  free probability theory. {I}.
\newblock {\em Comm. Math. Phys.}, 155(1):71--92, 1993.

\bibitem{MR2718125}
M.-K. von Renesse and M.~Scheutzow.
\newblock Existence and uniqueness of solutions of stochastic functional
  differential equations.
\newblock {\em Random Oper. Stoch. Equ.}, 18(3):267--284, 2010.

\bibitem{resflow}
P.~von Soosten and S.~Warzel.
\newblock The phase transition in the ultrametric ensemble and local stability
  of {D}yson {B}rownian motion.
\newblock Preprint available at ar{X}iv:1705.00923, 2017.

\end{thebibliography}
\bigskip
\begin{minipage}{0.5\linewidth}
\noindent Per von Soosten\\
Zentrum Mathematik, TU M\"{u}nchen\\
Boltzmannstra{\ss}e 3, 85747 Garching\\
Germany\\
\verb+vonsoost@ma.tum.de+
\end{minipage}%
\begin{minipage}{0.5\linewidth}
\noindent Simone Warzel\\
Zentrum Mathematik, TU M\"{u}nchen\\
Boltzmannstra{\ss}e 3, 85747 Garching\\
Germany\\
\verb+warzel@ma.tum.de+
\end{minipage}
\end{document}